\documentclass[11pt]{article}

\usepackage[letterpaper, portrait, margin=1in]{geometry}
\usepackage[usenames,dvipsnames,svgnames,table]{xcolor}
\usepackage[colorlinks=true,linkcolor=RawSienna,citecolor=ForestGreen]{hyperref}
\usepackage{algorithm}
\usepackage[noend]{algpseudocode}
\usepackage{url}
\usepackage{amsmath,amssymb,amsthm}
\usepackage{nicefrac}
\usepackage{thmtools,thm-restate}
\usepackage[noabbrev,capitalise,nameinlink]{cleveref}
\usepackage{mathtools}
\usepackage{xspace}
\usepackage{verbatim}
\usepackage{mathrsfs}
\usepackage{pgf}

\usepackage{todo}
\usepackage{tabularx}
\usepackage{enumitem}
\usepackage{tikz}
\usepackage{array}
\usepackage{dsfont}

\newcommand*\ie{i.\kern.1em e.\ }
\newcommand*\eg{e.\kern.1em g.\ }

\theoremstyle{plain}
\newtheorem{theorem}{Theorem}[section]
\newtheorem{lemma}[theorem]{Lemma}

\newtheorem{proposition}[theorem]{Proposition}
\newtheorem{claim}[theorem]{Claim}
\newtheorem{corollary}[theorem]{Corollary}
\newtheorem{question}[theorem]{Question}

\newtheorem{observation}[theorem]{Observation}

\theoremstyle{definition}

\newtheorem{definition}[theorem]{Definition}
\newtheorem{remark}[theorem]{Remark}
\newtheorem{example}[theorem]{Example}

\theoremstyle{plain}

\newcommand{\ignore}[1]{}

\DeclareMathOperator{\poly}{poly}

\newcommand{\dist}{\mathsf{dist}}

\newcommand{\define}{\vcentcolon=}

\newcommand{\inn}[1]{\langle #1 \rangle}

\newcommand{\zo}{\{0,1\}}

\newcommand{\col}{\mathsf{col}}

\newcommand{\fixedovalbox}[1]{%
  \tikz[baseline=(X.base)]\node(X)[draw, rounded corners, inner sep=3pt ]{#1};%
}

\newcommand{\domino}[2]{\small %
  \raisebox{-0.5em}{%
\fixedovalbox{$\stackrel{\stackrel{\scriptstyle #1}{\text{---}}}{\scriptstyle #2}$}}%
}

\newcommand{\cM}{\ensuremath{\mathcal{M}}}

\newcommand{\cP}{\ensuremath{\mathcal{P}}}
\newcommand{\cQ}{\ensuremath{\mathcal{Q}}}

\newcommand{\cY}{\ensuremath{\mathcal{Y}}}
\newcommand{\cW}{\ensuremath{\mathcal{W}}}
\newcommand{\cX}{\ensuremath{\mathcal{X}}}

\newcommand{\bN}{\ensuremath{\mathbb{N}}}

\newcommand{\bR}{\ensuremath{\mathbb{R}}}

\newcommand{\bZ}{\ensuremath{\mathbb{Z}}}

\newcommand{\sD}{\mathsf D}

\usepackage{fancybox}
\usepackage{soul}
\usepackage{parskip}

\newtoggle{anonymous}
\togglefalse{anonymous}

\newcommand{\GT}{\mathsf{GT}}

\newcommand{\CCC}{\mathsf{BPP}^0}
\newcommand{\BPP}{\mathsf{BPP}}
\newcommand{\EHD}{\mathsf{EHD}}
\newcommand{\THD}{\mathsf{THD}}
\newcommand{\IIP}{\mathsf{IIP}}

\title{No Complete Problem for Constant-Cost Randomized Communication}

\iftoggle{anonymous}{
\author{Anonymous Author(s)}
}{%
\author{
Yuting Fang \\ Ohio State University, USA \\ \texttt{fang.564@osu.edu}
\and
Lianna Hambardzumyan\thanks{Research partially supported by ISF grants 921/22 and 2635/19.} \\ Hebrew University of Jerusalem, Israel \\ \texttt{lianna.hambardzumyan@mail.huji.ac.il}
\and
Nathaniel Harms\thanks{Supported by an NSERC postdoctoral fellowship and the Swiss
State Secretariat for Education, Research, and Innovation (SERI) under contract number MB22.00026.} \\ EPFL, Switzerland \\ \texttt{nathaniel.harms@epfl.ch}
\and
Pooya Hatami\thanks{Supported by NSF grant CCF-1947546.} \\ Ohio State University, USA \\ \texttt{hatami.2@osu.edu}
}
}

\date{}

\begin{document}

\maketitle

\begin{abstract}
We prove that the class of communication problems with public-coin randomized constant-cost 
protocols, called $\CCC$, does not contain a complete problem. In other words, there is no
randomized constant-cost
problem $Q \in \CCC$, such that all other problems $P \in \CCC$
can be computed by a constant-cost \emph{deterministic} protocol with access to an oracle for
$Q$. We also show that the \textsc{$k$-Hamming Distance} problems form an infinite hierarchy within $\CCC$. Previously, it was known only that \textsc{Equality}
is not complete for $\CCC$.
We introduce a new technique, using Ramsey theory, that
can prove lower bounds against arbitrary oracles in $\CCC$, and more generally, we show that \textsc{$k$-Hamming Distance} matrices cannot be expressed as a Boolean combination of any constant number of matrices which forbid large \textsc{Greater-Than} subproblems.
\end{abstract}

\thispagestyle{empty}
\setcounter{page}{0}
\newpage

\section{Introduction}

One of the main goals in communication complexity is to understand the power of
randomized communication. The standard example is the \textsc{Equality} problem,
where two parties Alice and Bob are given strings $x,y \in \zo^n$, respectively,
and must decide if $x = y$. Given a shared source of randomness, Alice and Bob
can solve this problem with probability $\nicefrac 3 4$ using only 2 bits of
communication, regardless of input size, whereas a deterministic protocol
requires $n$ bits of communication. The \textsc{Equality} problem is therefore one of the most
extreme possible examples of the power of randomized communication, and to understand the power of randomness it is important to understand such extremes. For this purpose we define the class $\CCC$ of communication problems that, like
\textsc{Equality}, have constant-cost randomized public-coin protocols
(hereafter called merely \emph{constant-cost protocols}, see \cref{def:ccc}). The focused study of $\CCC$ was initiated by
\cite{HHH23eccc,HWZ22}, because:
\begin{itemize}[leftmargin=*, topsep=0pt, itemsep=0pt]

\item[-] There are many connections to other areas, including operator theory and Fourier analysis \cite{HHH23eccc}, learning \cite{LS09,FX14,HHPTZ22,HHM23,HZ24}, graph sparsity \cite{HWZ22,EHK22}, and implicit graph representations \cite{Har20,HWZ22,EHK22,EHZ23,HH22,NP23,HZ24}. 
    \item[-] Communication complexity is often applied to find lower bounds for
    other problems, and constant vs.~non-constant is the most basic lower bound
    question that one can ask, and yet it is often challenging --  several
    surprisingly non-trivial and natural communication problems are in $\CCC$
    (\eg computing small distances in planar graphs \cite{Har20,HWZ22,EHK22}, deciding incidence of
    certain low-dimensional point-halfspace arrangements \cite{HWZ22,HZ24},
    etc.). There are many lower bound techniques
    in the literature, but they often do not help answer questions about
    constant-cost communication, so we must develop new techniques (as in this
    paper).
    \item[-] Constant-cost
communication is a more ``fine-grained'' approach to understanding randomized
communication, which makes distinctions between different uses of randomness
(\eg public vs.~private, or \textsc{Equality} vs.~\textsc{Greater-Than}) that
are usually not differentiated, and allows for better understanding of ``dimension-free'' relations between
matrix parameters \cite{HHH23eccc}.
\item[-] If we wish to identify the structure of problems which allow for
efficient randomized communication, then we expect this structure to be most
evident in the constant-cost problems. Some open problems about randomized
communication  remain open even when restricted to their $\CCC$ versions,
including the size of monochromatic rectangles \cite{CLV19,HHH23eccc}, the role of
one- vs.~two-sided error and the existence of a complete problem.
Answering these questions for $\CCC$ is a first step towards the more general
answers. Constant-cost communication may be restrictive enough that one might
even hope to find a complete characterization of the problems in this class to
answer these questions.
\end{itemize}
See also the recent survey \cite{HH24} for more details.
Relevant to all of these motivations is the idea that $\CCC$ might contain a
\emph{complete} problem, \ie one ``truly randomized'' constant-cost protocol
$P$, such that all other constant-cost protocols can be rewritten as
deterministic protocols using $P$ as a subroutine (see
\cref{section:intro-reductions} for formal definitions). Identifying a complete
problem for $\CCC$ would answer almost all questions about $\CCC$ and provide a
nearly complete understanding of the most extreme examples of the power of
randomized communication. Earlier work \cite{HWZ22,HHH23eccc} proved that
\textsc{Equality} is \emph{not} complete for $\CCC$. We prove that there is
\emph{no} complete problem:

\begin{theorem}
\label{thm:intro-no-complete-problem}
There is no complete problem for $\CCC$.
\end{theorem}

We also prove the
following hierarchy of \textsc{$k$-Hamming Distance} problems in $\CCC$. The \textsc{$k$-Hamming Distance} problem
asks two players to decide whether the Hamming distance between $x,y \in \zo^n$ is $k$; for constant $k$, this is known to be 
in $\CCC$ (see \cref{section:intro-khd}).

\begin{theorem}
\label{thm:intro-hierarchy}
There are infinitely many constants $k$ such that \textsc{$k$-Hamming Distance} cannot be reduced to
\textsc{$(k-1)$-Hamming Distance}.
\end{theorem}

In particular, a simple application of our argument shows that this is true for $k \in \{1,2\}$,
recovering the result of \cite{HHH23eccc,HWZ22} that \textsc{$1$-Hamming Distance} does not
reduce to \textsc{Equality} (\ie \textsc{$0$-Hamming Distance}), while also separating 1- and
\textsc{2-Hamming Distance}, whereas previously it was not known whether \textsc{$1$-Hamming Distance}
is complete for $\CCC$.

\begin{theorem}
\label{thm:intro-2-vs-1}
\textsc{2-Hamming Distance} cannot be reduced to \textsc{1-Hamming Distance}.
\end{theorem}

This hierarchy echoes a similar hierarchy of \textsc{Integer Inner Product} functions within 
$\mathsf{BPP}$, established in \cite{CLV19}, though our proof is necessarily
very different. Those functions are denoted $\IIP_d$ for constant $d \in \bN$
and are defined on $dn$-bit integer vectors $x,y \in \bZ^d$ with $\IIP_d^n(x,y)
= 0$ if and only if $\inn{x,y} = 0$ (\cref{def:iip}). They are in the
communication complexity class $\BPP$ but are conjectured to have non-constant
cost (see \eg~\cite{CHHS23}). We use $\IIP_d$ as an example to show that, even
if a problem might have non-constant cost, and may therefore be ``more complex''
than any \textsc{$k$-Hamming Distance} problem, we can still easily separate
\textsc{$k$-Hamming Distance} from them using our technique:

\begin{theorem}
\label{thm:intro-iip}
For any constant $d$, there exists a constant $k$ such that \textsc{$k$-Hamming Distance} cannot
be reduced to $\IIP_d$.
\end{theorem}

To prove the theorems above, especially \cref{thm:intro-no-complete-problem}, it
is necessary to prove lower bounds on communication protocols with access to an oracle computing an arbitrary problem in $\CCC$.
There are many techniques in the literature for oracle lower bounds in communication, including
two lower bound techniques against the \textsc{Equality} oracle in $\CCC$ \cite{HHH23eccc,HWZ22,HZ24} and 
several techniques for lower bounds against \textsc{Equality} in
other communication complexity classes (\eg \cite{CLV19,PSW21,CHHS23,PSS23}).
However, none of these techniques have succeeded in proving separations within
$\CCC$ against oracles other than \textsc{Equality} -- not even against the
\textsc{$1$-Hamming Distance} oracle. An added challenge for proving lower
bounds against arbitrary oracles in $\CCC$ is that little is known about the
structure of problems in $\CCC$, including the basic question of whether they
have large monochromatic rectangles \cite{HHH23eccc}. We introduce a new
Ramsey-theoretic lower bound technique that gives separations against
\emph{arbitrary} oracles in $\CCC$. We give a proof overview and comparison to
prior work in \cref{section:intro-proof-overview}, after introducing definitions
in \cref{section:intro-prelim}. The proofs of the theorems above are in
\cref{section:applications}, and follow from a general lemma proved in
\cref{section:permutation-invariance} about the structure of oracle protocols
for computing \textsc{$k$-Hamming Distance}.

Our final result, proved in \cref{section:khd-submatrices}, is of a different type and deals with $\BPP$ reductions. An \emph{unbounded-size} $\BPP$ reduction, for a problem with $n$-bit inputs, allows $\poly\log n$-many oracle queries of arbitrary query size. It was proved in \cite{CLV19} (see also \cite{CHHS23}) that \textsc{Equality} is not complete for $\BPP$, because $\IIP_d^n$ (for $d \geq 3$) is irreducible to \textsc{Equality}, and more generally the $\IIP_d$ problems form an infinite hierarchy under unbounded-size $\BPP$ reductions.
This leaves open the question of how the $\IIP_d$ hierarchy interacts with \textsc{$k$-Hamming Distance}; in particular, whether $\IIP_d^n$ reduces to \textsc{$k$-Hamming Distance} for $k = \poly\log n$ under unbounded-size reductions. We prove this is not so, by
showing certain ``dimension-free'' relations: the \textsc{Equality} oracle complexity and $\gamma_2$-norm
of arbitrary $N \times N$ submatrices of \textsc{$k$-Hamming Distance} depends on $N$ but \emph{not} on
the underlying dimension. This implies:

\begin{theorem}[Informal]
\label{thm:intro-bpp}
    For every $d \geq 3$, $\IIP_d^n$ does not reduce to \textsc{$k$-Hamming Distance}, under unbounded-size $\BPP$ reductions, for any $k = k(n) \leq \nicefrac{n}{(\log n)^{\omega(1)}}$.
\end{theorem}

A careful examination of \cite{CLV19} gives the incomparable statement that, for every constant $k$, there exists an unspecified constant $d = d(k) \geq 6$ such that $\IIP_d^n$ requires $\Theta(n)$ \textsc{$k$-Hamming Distance} oracle queries. Our lower bound is $\Omega(\nicefrac{n}{k\log n})$ queries, but it applies to $d=3$ and non-constant $k$.

We conclude the paper with a discussion and open problems in \cref{section:discussion}.

\subsection{Preliminaries: Reductions, \texorpdfstring{$k$}{k}-Hamming Distance, and Greater-Than} 
\label{section:intro-prelim}

Let us now formalize the notions of reductions and completeness within $\CCC$, and state
some required facts about the \textsc{$k$-Hamming Distance} and \textsc{Greater-Than} problems.

\subsubsection{Constant-Cost Reductions}
\label{section:intro-reductions}

A \emph{communication problem} $\cP$ is a sequence $\cP = (P_N)_{N \in \bN}$ where $P_N \in \zo^{N
\times N}$ is an $N \times N$ Boolean matrix. We will use $N$ to denote the size of the matrix and,
when it is natural, we may define another parameter (\eg sometimes writing $n$ for the number of bits in the input). For any Boolean matrix $M$, we
write $\mathsf{R}(M)$ for the minimum cost of a two-way public-coin randomized communication
protocol computing $M$ with success probability at least $2/3$, and we refer to standard texts \cite{KN96,RY20} for an introduction to randomized communication. For a communication problem $\cP$,
we write $\mathsf{R}(\cP)$ for the function $N \mapsto \mathsf{R}(P_N)$.

\begin{definition}[$\CCC$]
\label{def:ccc}
A problem $\cP$ has \emph{constant cost} if there exists a constant $c$ such that $\mathsf{R}(\cP)
\leq c$, \ie for all $N \in \bN$, it holds that $\mathsf{R}(P_N) \leq c$. We define $\CCC$ as the
set of all problems $\cP$ which have constant cost.
\end{definition}

Our notation follows in spirit the notation for complexity classes like $\mathsf{AC}^i$ and
$\mathsf{NC}^i$: For any $i$, we think of $\BPP^i$ as being the class of communication problems $\cP$ with $\mathsf{R}(\cP) = O(\log^i \log N)$, \ie $O(\log^i n)$ where $n = \lceil \log
N \rceil$ is the number of bits required to represent the inputs. Therefore the standard
communication complexity class $\BPP$ is $\BPP = \bigcup_{i=0}^\infty \BPP^i$.

\begin{remark}
Unlike $\BPP$, the class $\CCC$ remains unchanged
regardless of whether it is defined in terms of two-way, one-way, or simultaneous communication
protocols \cite{HWZ22}, but it is \emph{not} equivalent to replace public randomness with private randomness.
\end{remark}

We must now define communication protocols with oracle queries. It is common to study communication
with oracles (see \eg \cite{BFS86,GPW18,CLV19,PSW21,CHHS23,PSS23}), but for $\CCC$ the most natural
definition of oracle queries is different from the standard one. Usually, the size of the input to
the oracle query is restricted -- on $n$-bit inputs, the oracle should be queried only on
$\poly(n)$-bit inputs,  because making $\poly\log n$ queries to problems with complexity $\poly\log
m$ on query inputs of $m = \poly(n)$ bits will produce a protocol of complexity $\poly\log n$,
preserving the usual notion of ``efficiency''. For $\CCC$, the oracles should allow queries of
\emph{arbitrary} size, for the same reason that this preserves our notion of efficiency.  This leads
to the following definition which is implicit in prior works \cite{HWZ22,EHK22,HHH23eccc} and
explicit in \cite{HZ24}:

\newcommand{\QS}{\mathsf{QS}}
For any set $\cM$ of Boolean matrices, we will define
the \emph{query set} $\QS(\cM)$ of $\cM$ as the set of all matrices $Q$ obtained from matrices in
$\cM$ by the following operations:
\begin{enumerate}
\item Arbitrary row and column permutations;
\item Taking arbitrary submatrices; and
\item Duplicating arbitrary rows or columns.
\end{enumerate}
The difference between our reductions and the standard reductions for
communication complexity comes from Item (2). Problems in $\CCC$ are hereditary,
in a way that problems in standard $\BPP$ are not: if $\cP$ has a randomized
protocol with constant cost $c$, then any matrix $P \in \QS(\cP)$ also has a
randomized protocol with cost $c$ (because $P$ is obtained by choosing a problem
$P' \in \cP$, taking a submatrix $P''$ of $P'$, which cannot increase the
communication cost, and then duplicating rows and columns and permuting them,
which does not change the communication cost). On the other hand, a problem $\cP$ in the standard
$\BPP$ class with complexity $\poly(\log\log N)$ can have matrices $P \in \QS(\cP)$ with cost
$\Theta(\log N)$, as in the following example:

\begin{example}
The \textsc{$\log(n)$-Hamming Distance} problem, defined on binary strings
$\zo^n$, has cost $\Theta( \log(n) \log\log(n) ) = \Theta(\log\log (N)
\log\log\log (N))$ on matrices of size $N \times N = 2^n \times 2^n$ (which
follows from optimal bounds on the communication complexity of
\textsc{$k$-Hamming Distance} for non-constant $k = \log n$
\cite{HSZZ06,Sag18}). But since the VC dimension of the \textsc{$k$-Hamming
Distance} matrices depends on $k$ (see \cref{def:vc-dimension,prop:two-sig-vc}),
the set $\QS$ contains \emph{all} matrices.
\end{example}

\begin{definition}[Communication with Oracle Queries]
Let $\cM$ be any set of Boolean matrices and let $P \in \zo^{N \times N}$. Then we write
$\mathsf{D}^{\cM}(P)$ for the minimum cost of a deterministic communication protocol with access to
an oracle for $\cM$, defined as follows. A \emph{communication protocol} with oracle access to $\cM$ is a binary tree $T$ with
inner nodes $V$ and leaves $L$. Each inner node $v \in V$ is associated with an $N \times N$ matrix
$Q_v \in \QS(\cM)$ and each leaf $\ell \in L$ is associated with an output bit $b_\ell \in \zo$.
An input $x,y\in [N]$ then naturally defines a path from the root to a leaf $\ell\in
L$, wherein at each node the $(x,y)$ entry of the corresponding query matrix is used to decide
whether to travel left or right. The output of the protocol is then the label of the reached leaf,
$b_\ell$.
\end{definition}

For simplicity of the definition, we force every round of communication to be via an oracle query -- the players cannot send messages directly to each other.
As observed by \cite{CLV19}, a standard round of communication can be simulated by an oracle query as long as the set $\cM$ is non-trivial, \ie it
contains at least one of the matrices $\big[\begin{smallmatrix}1 & 0 \\ 0 &
1\end{smallmatrix}\big]$, $\big[\begin{smallmatrix}1 & 1 \\ 0 & 1\end{smallmatrix}\big]$, or their
permutations or Boolean negations.

The following proposition is easy to prove using standard
error-boosting techniques:

\begin{proposition}
Let $\cP, \cQ$ be communication problems such that $\cQ \in \CCC$ and
$\mathsf{D}^{\cQ}(\cP) = O(1)$. Then $\cP \in \CCC$.
\end{proposition}

Then the following is the most natural definition of reductions within $\CCC$.

\begin{definition}[Constant-Cost Reductions]
\label{def:reductions}
Let $\cP, \cQ$ be communication problems. Then we say that $\cP$ has a \emph{constant-cost
reduction} to $\cQ$ (or simply it \emph{reduces to $\cQ$}) if $\mathsf{D}^\cQ(\cP) = O(1)$.
\end{definition}

We now have the natural notion of completeness for $\CCC$:

\begin{definition}[Completeness in $\CCC$]
A communication problem $\cQ$ is \emph{complete} for $\CCC$ if $\cQ \in \CCC$ and, for all $\cP \in
\CCC$, $\cP$ is constant-cost reducible to $\cQ$.
\end{definition}

We will use an equivalent and often more convenient definition of constant-cost
reductions in terms of Boolean combinations of query matrices. We refer the
reader to \cite{HZ24} for a simple proof\footnote{The idea is that
$f(Q_1(x,y), \dotsc, Q_c(x,y))$ is the function that simulates the communication
protocol using the answers to all of the queries.}. 

\begin{proposition}
\label{prop:functional-rank}
Let $\cM$ be a set of Boolean matrices and let $\cP$ be a communication problem. Then
$\mathsf{D}^\cM(\cP) = O(1)$ if and only if there exists a constant $c \in \bN$ and a function $f :
\zo^c \to \zo$ such that, for all $N \in \bN$, there exist $Q_1, \dotsc, Q_c \in \QS(\cM)$ such that
\begin{equation}
\label{eq:constant-cost-reduction}
  \forall x,y \in [N] :\qquad P_N(x,y) = f(Q_1(x,y), Q_2(x,y), \dotsc, Q_c(x,y)) \,.
\end{equation}
\end{proposition}
\begin{remark}
    The above proposition claims that we have a ``uniform'' function $f : \zo^c \to \zo$, that works for \emph{every} problem size $N$.
    Swapping the quantifiers to allow a different
    function $f_N : \zo^c \to \zo$ for each $N$ does not increase the power of the protocol, because there is only a
    constant number of functions $\zo^c \to \zo$,
    so the choice of $f_N$ can be
    encoded\footnote{The first $r$ queries can
    be used to select a function $f_N$
    out of a space of $2^r$ functions} in $f$.
\end{remark}

Constant-cost reductions are also natural and useful in the study of \emph{implicit graph
representations}; see
\cite{Har20,HWZ22,EHK22,NP23,HZ24} for more on this connection, and \cite{Chan23} for reductions
motivated directly from implicit graph representations.

\subsubsection{\texorpdfstring{$k$}{k}-Hamming Distance}
\label{section:intro-khd}

Let $\dist(x,y)$ denote the Hamming distance between two strings.
There are two variations of the \textsc{$k$-Hamming Distance} problem. The \textsc{Exact $k$-Hamming
Distance} problem is defined as $\EHD_k = (\EHD^n_k)_{n \in \bN}$ where
\[
    \forall x,y \in \zo^n : \qquad \EHD_k^n(x,y) = 1 \text{ if and only if } \dist(x,y) = k \,.
\]
The \textsc{Threshold $k$-Hamming Distance}
problem is defined as $\THD_k = (\THD_k^n)_{n \in \bN}$ where
\[
    \forall x,y \in \zo^n : \qquad \THD_k^n(x,y) = 1 \text{ if and only if } \dist(x,y) \leq k \,.
\]
Observe that $\THD_0 \equiv \EHD_0$ is the \textsc{Equality} problem. For other
constant values of $k\geq 1$, the problems are also equivalent under
constant-cost reductions: it is easy to show that $\mathsf{D}^{\THD_k}(\EHD_k)
\leq 2$ (since $\EHD_k \equiv \THD_k \wedge \neg \THD_{k-1}$, and $\THD_{k-1}$
can be computed by one query to $\THD_k$ by padding the input), and
$\mathsf{D}^{\EHD_k}(\THD_k) \leq k$ (since $\THD_k \equiv \bigvee_{t = 0}^k
\EHD_t$ and $\EHD_t$ can be computed by one query to $\EHD_k$ by padding the
input). The two-way public-coin randomized communication cost of these problems
is $O(k \log k)$ \cite{Yao03,HSZZ06} (with a matching lower bound when $k <
\sqrt n$ \cite{Sag18}) so for every constant $k$, $\EHD_k$ and $\THD_k$ are in
$\CCC$.

\subsubsection{Greater-Than and Stability}

The \textsc{Greater-Than} problem is $\GT = (\GT_t)_{t \in
\bN}$, where the matrix $\GT_t$ is defined on $i,j \in [t]$ as
\[
  \GT_t(i,j) = 1 \text{ if and only if } i \leq j \,.
\]
Following the terminology of \cite{HWZ22}, we say a problem $\cP$ is
\emph{stable} if the largest \textsc{Greater-Than} subproblem within $\cP$ has constant size.
Formally:
\begin{definition}[Stability]
\label{def:stable}
A set $\cM$ of matrices is \emph{stable} if there exists a constant $t$ such that $\GT_t \notin \QS(\cM)$.  
Equivalently, $\cM$
is stable if there exists a constant $t$ such that, for any matrix $M \in \cM$, and any set of rows
$x_1, \dotsc, x_m$ and columns $y_1, \dots, y_m$ of $M$ which satisfy $M(x_i, y_j) = 1$ iff $i \leq
j$, it holds that $m \leq t$.
\end{definition}
It is equivalent to require $\neg \GT_t \notin \cM$ instead of $\GT_t \notin \cM$, where $\neg \GT_t$ denotes
Boolean negation, because $\neg \GT_t$ is a submatrix of $\GT_{t+1}$.
We will require the following observation, which follows from the known fact
that $\mathsf{R}(\GT_t) = \Theta(\log\log t)$ \cite{Nis93,Vio15,RS15,BW16} and
therefore $\GT \notin \CCC$.
\begin{observation}
\label{prop:ccc-stable}
Every problem $\cQ \in \CCC$ is stable.
\end{observation}

\begin{remark}
    Stability is a necessary but not sufficient condition for a problem $\cQ$ to
    belong to $\CCC$. For example, any problem $\cQ \in \CCC$ must contain at
    most $2^{O(N \log N)}$ $N \times N$ matrices \cite{HWZ22}, but the family of
    all $K_{2,2}$-free matrices (\ie matrices with no $2 \times 2$ rectangle of
    1s), which is stable, contains more matrices than this \cite{LZ15}. Even
    restricting to problems $\cQ$ that are both stable and have at most $2^{O(N
    \log N)}$ $N \times N$ matrices is insufficient to guarantee membership in
    $\CCC$ \cite{HHH22counter}.
\end{remark}

\begin{remark}
    The size of the largest \textsc{Greater-Than} inside a matrix $M$ is also
    called the Littlestone dimension, which characterizes the number of mistakes
    made by an online learning algorithm \cite{Lit88,ALMM19}. Any stable set of
    matrices describes a hypothesis class that is learnable in a bounded
    number of mistakes, while $\CCC$ is the family of hypothesis
    classes that are learnable in a bounded number of mistakes with the
    \emph{perceptron} algorithm, due to the bounded-margin embedding of
    \cite{LS09}.
\end{remark}

\subsection{Proof Overview and Comparison to Prior Work}
\label{section:intro-proof-overview}

\subsubsection{Prior Techniques}
\newcommand{\EQ}{\textsc{Eq}}
We wish to prove lower bounds for the $\EHD_k$ problem, against arbitrary oracles.
By \cref{prop:functional-rank}, if we assume $\EHD_k$ reduces to $\cQ$, we can write,
for all $x,y \in \zo^n$,
\begin{equation}
\label{eq:proof-overview}
  \EHD_k^n(x,y) = f(Q_1(x,y), Q_2(x,y), \dotsc, Q_c(x,y)) \,,
\end{equation}
where $Q_1, \dotsc, Q_c \in \QS(\cQ)$.
A natural approach to show that $\EHD_k$ cannot be reduced to $\cQ$ is to define a complexity
measure $\kappa$ which is bounded on $\cQ$, such that, say, $\kappa\left( f(Q_1, \dotsc, Q_c)
\right) \leq g(\kappa(Q_1), \dotsc, \kappa(Q_c))$, for some function $g$ independent of $n$, while
$\kappa(\EHD_k^n) = \omega(1)$. This is the approach taken in several prior works~\cite{HHH23eccc, CLV19,
CHHS23, PSS23} to prove lower bounds against the \textsc{Equality} oracle.
The $\gamma_2$
norm, used in \cite{HHH23eccc, CHHS23}, cannot separate $\EHD_k$ from $\EHD_1$, for any constant $k > 1$.
Another measure, \emph{$\eta$-area}, was introduced in \cite{CLV19} to show separations within $\BPP$. It is unclear whether this could be used to show separations within $\CCC$, which would require a technical analysis of the monochromatic rectangles
within all submatrices of $\EHD_k$, and indeed all problems in $\CCC$, whereas the rectangle analyses in \cite{CLV19} fail for $\EHD_1$. It is not even known whether every problem in $\CCC$ has large (linear-size) monochromatic rectangles  \cite{HHH23eccc}, which would be the first step in proving \cref{thm:intro-no-complete-problem} using the technique of \cite{CLV19}.

A more structural (but non-quantitative) approach was taken in \cite{HWZ22,HZ24}, which depended
fundamentally on the fact that the \textsc{Equality} oracle
partitions the inputs into very simple monochromatic rectangles. Every step of that proof fails when
\textsc{Equality} is replaced with $\EHD_1$. The challenge with a structural approach is that it is
difficult to understand the structure of an arbitrary Boolean combination of matrices $Q_1, \dotsc,
Q_c$, even if the structure of $Q_1, \dotsc, Q_c$ are themselves well-understood, and the structure
of matrices belonging to problems in $\CCC$ is \emph{not} well-understood.

\subsubsection{Proof Overview}
We overcome these challenges in a way that is conceptually
simpler than the previous bounds against only the \textsc{Equality} oracle, and, unlike the prior work, does not involve any argument about monochromatic rectangles. Observe that
$\EHD_k$ is permutation-invariant in the following way. For every
pair of inputs $x,y \in \zo^n$, we think of $(x,y)$ as defining a sequence of ``dominoes''\!, \ie
pairs $ab \in \zo^2$, where $x_1y_1$ is the first domino, $x_2y_2$ is the second, and so on:
\begin{align*}
  \raisebox{-0.5em}{$\stackrel{\stackrel{\scriptstyle x}{\phantom{-}}}{\scriptstyle y}$}
  = \domino{x_1}{y_1} \, \domino{x_2}{y_2} \dots \domino{x_n}{y_n} 
\end{align*}
Then the output of $\EHD_k$ is invariant under permutations on these dominoes.

Our next ingredient is the
basic observation that every problem $\cQ$ in $\CCC$ has a fixed constant $t$ such that no
\textsc{Greater-Than} communication problem of size larger than $t \times t$ exists in $\QS(\cQ)$,
\ie problems $\cQ \in \CCC$ are ``stable'' (\cref{prop:ccc-stable}).

The function $\EHD_k(x,y) = f(Q_1(x,y), \dotsc, Q_c(x,y))$
is, as a whole, invariant under ``domino permutations'' on the input, but \emph{a priori}
we have no similar guarantee on 
the queries $Q_i$. Our goal is to show that each
query $Q_i$ must \emph{also} be permutation-invariant. We accomplish this (in \cref{section:proof}) by thinking of the query
responses as a coloring of a hypergraph whose vertices are the coordinates $[n]$ of the input.
Using only the permutation invariance of $\EHD_k$ and stability of the queries, we apply the
hypergraph Ramsey theorem in stages, in each stage increasing the number of permutations under which
the \emph{queries} $Q_i$ are invariant, until we achieve our goal.

From here, we see that, if there are two classes $A$ and $B$ of inputs $(x,y)$, where $A$ and $B$
are equivalence classes under domino permutations, and furthermore the output of $\EHD_k$ is
different on inputs $A$ than on inputs $B$, then there must be a query
$Q \in \QS(\cQ)$ that distinguishes \emph{all} inputs in $A$ from \emph{all} inputs in $B$. In this way,
we transform the
task of finding a lower bound for $\EHD_k$ against any constant number of arbitrary queries in $\QS(\cQ)$,
into the task of finding a lower bound for a partial subproblem of $\EHD_k$ against a
\emph{single} query from $\QS(\cQ)$, which is done in \cref{section:applications}.

Our main idea, forcing the algorithm to behave a certain way using Ramsey theory, was unexpectedly inspired by unrelated works
in \emph{property testing}, which use a more direct application of Ramsey theory to force 
testing algorithms to process random samples in a certain way \cite{Fis04,DKN15}.

\section{Permutation Invariance of \texorpdfstring{$k$}{k}-Hamming Distance Protocols}
\label{section:permutation-invariance}

We now prove the main lemma, which shows that any constant-cost reduction
from $\EHD_k$ to an arbitrary stable set of matrices $\cM$ (\cref{def:stable}) can be forced to use oracle queries that
are invariant under permutations on the input.

\begin{definition}[Permutation Invariance]
For a matrix $Q : \zo^n \times \zo^n \to \zo$, we say that $Q$ is \emph{permutation-invariant} if
for every $x, y \in \zo^n$ and every permutation $\pi : [n] \to [n]$,
\[
  Q(x,y) = Q(x_\pi, y_\pi) \,,
\]
where $x_\pi \define (x_{\pi(1)}, x_{\pi(2)}, \dotsc, x_{\pi(n)})$, and $y_\pi$ is defined
similarly.
\end{definition}

\begin{lemma}
\label{lem:permutation_invariant}
Let $k \in \bN$ and let $\cQ$ be any stable set of Boolean matrices. Suppose $\sD^\cQ(\EHD_k) = O(1)$. Then there exists a
constant $c$ and a function $f : \zo^c \to \zo$,  such that, for all
$n \in \bN$, there are $c$
\emph{permutation-invariant} queries $Q_1, \ldots, Q_c \in \QS(\cQ)$
satisfying
\[
  \forall x,y \in \zo^n \quad  \EHD_k^n(x,y) = f(Q_1(x,y), \dotsc, Q_c(x,y)) .
\]
\end{lemma}

We state some notation and the Ramsey theorem in \cref{section:proof-setup} and prove
the lemma in \cref{section:proof}.

\subsection{Setup: Dominoes and the Hypergraph Ramsey Theorem}
\label{section:proof-setup}

\begin{definition}(Domino)
We call a pair $ab \in \zo^2$ a \emph{domino} and we denote it as $\domino{a}{b}$.
For any $n \in \bN$ and for any pair $(x,y) \in \zo^n \times \zo^n$, the \emph{dominoes} of $(x, y)$ is the sequence 
\[
\domino x y \define \left( \domino{x_1}{y_1}, \domino{x_2}{y_2}, \dotsc, \domino{x_n}{y_n}\right).
\]
For a set of dominoes $\Delta \subseteq \zo^2$, we denote the complement of $\Delta$ by $\overline{\Delta} = \zo^2 \setminus \Delta$.
\end{definition}

\newcommand{\type}[1]{\left\langle #1 \right\rangle}
\begin{definition}[Type] 
\label{def:signature}
Let $\Delta \subseteq \zo^2$ be a set of dominoes. A \emph{$\Delta$-type}
is a tuple $\type{\Gamma_\Delta, \tau}$ containing a \emph{$\Delta$-signature} $\Gamma_\Delta \in \Delta^*$,
which is an ordered sequence of dominoes in $\Delta$, and a \emph{tally} $\tau = [\tau_{ab}]_{a,b \in \zo}$, which is a sequence with $\tau_{ab} \in \bZ$.
For a pair $(x,y) \in \{0,1\}^n \times \{0,1\}^n$ and a set of dominoes $\Delta \subseteq \zo^2$, the
\emph{$\Delta$-type} of $(x,y)$, denoted $\chi_\Delta(x,y)$,
is the tuple $\type{\Gamma_\Delta(x,y), \tau(x,y)}$, where 
\begin{itemize}
    \item $\Gamma_\Delta(x,y)$ is the $\Delta$-signature of $(x,y)$: the subsequence 
    of dominoes of $(x,y)$ that belong to $\Delta$,
    \[
        \Gamma_\Delta(x,y) = \left( \domino{x_i}{y_i} \mid \text{ for all } i \in [n], \hspace{0.5em}  \domino{x_i}{y_i}
            \in \Delta\right);
    \]
    \item $\tau(x,y) = [\tau_{ab}(x,y)]_{a,b \in \zo}$ denotes the \emph{tally} of the dominoes of $(x,y)$, where $\tau_{ab}(x,y)$ is the number of times $\domino a b$ occurs in the dominoes of $(x,y)$.
\end{itemize}
For example, with $\Delta = \left\{\domino{0}{1},\domino{1}{0}\right\}$:
\[
    \chi_\Delta\left(\domino{0110000}{0101001}\right) = \type{\Gamma_\Delta=\left(\domino{1}{0}, \domino{0}{1}, \domino{0}{1}\right), [\tau_{00}= 3,\tau_{01}=2, \tau_{10}=1 , \tau_{11}= 1] }.
\]
\end{definition}

\begin{definition}[Shuffle Invariance]
For any matrix $Q : \zo^n \times \zo^n \to \zo$ and a set of dominoes $\Delta \subseteq \zo^2$, we
say that $Q$ is \emph{$\Delta$-shuffle invariant} if
\[
    \forall x,y,u,v \in \zo^n : \qquad \chi_{\overline \Delta}(x,y) = \chi_{\overline \Delta}(u,v) \implies Q(x,y) = Q(u,v) \,.
\]
In other words, $Q$ is \emph{$\Delta$-shuffle invariant} if its value on $(x,y)$ is invariant under any
swap of consecutive dominoes in the sequence 
$\domino{x_1}{y_1} \domino{x_2}{y_2} \dotsm \domino{x_n}{y_n}$ which involve a domino in $\Delta$; any permutation of the dominoes achieved by a sequence of such swaps preserves the relative order of dominoes \emph{outside} $\Delta$,
and therefore preserves the $\overline{\Delta}$-type. 
$Q$ is permutation-invariant if it is $\Delta$-shuffle invariant for the full set of dominoes $\Delta = \zo^2$.
\end{definition}

We require the well-known hypergraph Ramsey theorem. For any set $T$ and any $0 \leq t \leq |T|$, write ${T \choose t}$ for the family of subsets of $T$ of cardinality $t$.

\begin{theorem}[Hypergraph Ramsey theorem]
\label{thm:hypergraph-ramsey}
For any $\alpha,\beta\in \bN$ and $\sigma \geq \alpha$, there exists $R = R(\alpha,\beta,\sigma)$ such that for any coloring $\kappa :\binom{[R]}{\alpha} \to [\beta]$, there exists a subset $T \subseteq [R]$ of size $\sigma$ such that $\kappa$ is constant on ${T \choose \alpha}$.
\end{theorem}
We use an easy corollary of this theorem and provide a proof for the sake of completeness. For any set $T$ and any $0 \leq t \leq |T|$, write
${T \choose \leq t}$ for the set of subsets of $T$ of cardinality at most $t$.

\begin{corollary}
\label{cor:hypergraph-ramsey}
For any $\alpha,\beta \in \bN$ and $\sigma \geq \alpha$, there exists $N = N(\alpha,\beta,\sigma)$ such that for any coloring $\kappa : { [N] \choose \leq \alpha } \to [\beta]$, there exists $T \subseteq [N]$ of size $\sigma$ such that $\kappa$ is constant on ${T \choose \alpha'}$ for every $\alpha' \leq \alpha$.
\end{corollary}
\begin{proof}
    For any $\alpha',\beta',\sigma' \in \bN$, write $R(\alpha',b',\sigma')$ for the number obtained from \cref{thm:hypergraph-ramsey}. 

    We prove the statement by induction on $\alpha$. For $\alpha=1$ the conclusion is easy to obtain. Now assume $\alpha > 1$ and write $M \define N(\alpha-1,\beta,\sigma)$ for the number obtained by induction for parameters $\alpha-1,\beta,\sigma$. We define $N \define N(\alpha,\beta,\sigma) \define R(\alpha,\beta, N(\alpha-1,\beta,\sigma))$.
    
    Let $\col : { [N] \choose \leq \alpha } \to [\beta]$.
    Let $\col_\alpha : { [N] \choose \alpha } \to [\beta]$ be the function $\col$ restricted to domain ${ [N] \choose \alpha }$.
    Then by \cref{thm:hypergraph-ramsey}, there exists a set $T \subseteq [N]$ of size $M$ such that $\col_\alpha$
is constant on domain ${ T \choose \alpha }$. Relabel the elements of $T$ as $[M]$ and define the
function $\col' : { [M] \choose \leq \alpha-1 } \to [\beta]$ as the function $\col$ restricted to the domain ${T \choose \leq \alpha-1}$ with the elements of $T$ relabeled as $[M]$.
    $T$ has cardinality $M = N(\alpha-1,\beta,\sigma)$, so by induction $\col'$ is constant on ${ [M] \choose \alpha'}$
for each $\alpha' \leq \alpha-1$. Then $\col$ is constant on ${ T \choose \alpha'}$ for $\alpha' \leq \alpha-1$, and since
$\col_\alpha$ is constant on ${T \choose \alpha}$, this implies the conclusion.
\end{proof}

\subsection{Proof of \texorpdfstring{\cref{lem:permutation_invariant}}{Structural Lemma}}
\label{section:proof}

We now prove the permutation invariance lemma. Let $\cQ$ be any
stable set of matrices.
For any set $\Delta \subseteq \zo^2$ of dominoes, consider the following 
statement:

\noindent
\textbf{$\Delta$-Shuffle Property:}
\emph{There exist a constant $c$ and a function 
$f : \zo^c \to \zo$ such that for every $n \in \bN$
there are $Q_1, \dotsc, Q_c \in \QS(\cQ)$ such that
\begin{equation}
\label{eq:invariance}
    \forall x,y \in \zo^n : \qquad  
    \EHD_k^n(x,y) = f(Q_1(x,y), \dotsc, Q_c(x,y)) \,,
\end{equation}
and each $Q_i$ is $\Delta$-shuffle invariant.}

\cref{prop:functional-rank} guarantees that the $\Delta$-shuffle property 
holds for $\Delta = \emptyset$. Our goal is to show that it holds for $\Delta = \zo^2$.

\begin{claim}
\label{claim:main-step-1}
Let $\Delta \subseteq \zo^2$ be any set of dominoes, and suppose the
$\Delta$-shuffle property holds. Then for any $a \in \zo$,
the $\left(\Delta \cup  \left\{\domino a a \right\}\right)$-shuffle property also holds.
\end{claim}
\begin{proof}[Proof of claim]
Let $\Delta' = \Delta \cup \left\{\domino a a\right\}$, and
let $\overline a$ denote the negation of the bit $a$. Then, let $D = \left\{\domino{\overline{a}}{\overline{a}},\domino{1}{0},\domino{0}{1}\right\}$. For $n \in \bN$, let $N = N(n, 2^b, n)$ be the number obtained from \cref{cor:hypergraph-ramsey}, where $b = c \cdot 3^n$, which will be justified below. We will embed $\zo^n$ into the larger domain of $\zo^N$ such that the embedding preserves the Hamming distance and the $D$-type of any two strings, and allows to show the $\Delta'$-shuffle invariance of queries. 

The first two properties are easy to satisfy. Take any subset of coordinates $T \subset [N]$ of size $|T|=n$, and let $\phi_T : \zo^n \to \zo^N$  be the map that writes $x \in \zo^n$ into the coordinates of $T$ in the order-preserving way and sets the coordinates outside of $T$ to be $a$. Observe that, for all $x,y \in \zo^n$,
\[
   \dist(x,y) = \dist(\phi_T(x), \phi_T(y)) \text{ and } \Gamma_{D} \left(x,y \right) = \Gamma_{D} \left(\phi_T(x),\phi_T(y)\right) \,.
\]
Next, we choose a set $T$ that helps us to show the $\Delta'$-shuffle invariance of the queries.
By assumption, there exists $f : \zo^c \to \zo$ and
$Q'_1, \dotsc, Q'_c : \zo^N \times \zo^N \to \zo$ with each $Q'_i \in \QS(\cQ)$ being
$\Delta$-shuffle invariant, such that
\[
\forall X,Y \in \zo^N : \qquad \EHD_k^N(X,Y) = f(Q'_1(X,Y), \dotsc, Q'_c(X,Y)) .
\]
where we write $X, Y \in \zo^N$ to distinguish them from lower-dimensional $x,y \in \zo^n$.

To each $S \subseteq [N]$ with $|S|\leq n$ we assign a color, which is a binary string $\col(S) \in \zo^b$, as follows.
Let $s = |S|$ and define the color $q(d)$ of a domino vector $d \in
D^s$ to be the sequence
of $c$ bits $q(d) = (Q'_1(U,V), \dotsc, Q'_c(U,V))$, where $U,V \in \zo^N$ is the unique pair whose
$D$-signature is $d$ and the dominoes of $d$ are written in the coordinates $S$ of $[N]$. 
Now set $\col(S)$ to be $(q(d))_{d \in D^s}$, the concatenation of the colors $q(d)$ of
all possible signature vectors $d \in
D^s$, in lexicographic order of $d$.  The total number
of bits in $\col(S)$ is at most $c \cdot 3^s \leq c \cdot 3^n = b$, so there are at most $2^b$ colors.

By \cref{cor:hypergraph-ramsey}, there exists a set $T \subseteq [N]$ of size $|T|=n$ such that for every $s \leq n$, the subsets $S \subseteq T$ of cardinality $|S|=s$ each have the same color $\col(S)$. Let $\phi\define \phi_T$ be the map defined above, which preserves the Hamming distance of $x,y \in \zo^n$ and their $D$-signature.
For each $i \in [c]$, we now define the matrix $Q_i : \zo^n \times \zo^n \to \zo$ by $Q_i(x,y) \define Q'_i(\phi(x), \phi(y))$. Observe that $Q_i$ is a submatrix of $Q'_i$, so $Q_i \in \QS(\cQ)$ since $\QS(\cQ)$ is closed
under taking submatrices. Now,
\begin{align*}
    \forall x,y \in \zo^n : \qquad f(Q_1(x,y), \dotsc, Q_c(x,y))
        &= f(Q_1'(\phi(x), \phi(y)), \dotsc, Q_c'(\phi(x),\phi(y))) \\
        &= \EHD_k^N(\phi(x), \phi(y)) \\
        &= \EHD_k^n(x,y) \,.
\end{align*}
It remains to show that each $Q_i$ is $\Delta'$-shuffle invariant. Observe that
each $Q_i$ is $\Delta$-shuffle invariant because it is a submatrix of $Q'_i$ and
$Q'_i$ is $\Delta$-shuffle invariant. Indeed, for any $x,y,u,v \in \zo^n$ and
any $i \in [c]$,
\begin{align*}
\qquad  \chi_{\overline \Delta}(x,y) = \chi_{\overline \Delta}(u,v) 
&\implies 
\chi_{\overline \Delta}(\phi(x),\phi(y)) = \chi_{\overline \Delta}(\phi(u),\phi(v)) \\
&\implies 
Q'_i(\phi(x),\phi(y)) = Q_i(\phi(u),\phi(v)) \\
&\implies
Q_i(x,y) = Q_i(u,v) \,.
\end{align*}
Now let $x,y,u,v \in \zo^n$ such that $ \chi_{\overline{\Delta'}}(x,y) =
\chi_{\overline{\Delta'}}(u,v)$. We must show that $Q_i(x,y) = Q_i(u,v)$. First,
assume that $(x,y)$ and $(u,v)$ have the same $D$-signature, so that the
dominoes of $(u,v)$ are obtained from those of $(u,v)$ by a sequence of swaps of
consecutive dominoes, where each swap involves an $\domino{a}{a}$ domino, (so
that their subsequences of non-$\domino a a$ dominoes are the same). Then the
sets 
\[
    S_1 \define \{ i \in [N] \mid \phi(x)_i =\overline{a} \text{ or } \phi(y)_i =\overline{a} \}\qquad \text{and} \qquad S_2 \define \{ i \in [N] \mid \phi(u)_i =\overline{a} \text{ or } \phi(v)_i =\overline{a}\}
\]
have the same size, thus also the same color $\col(S_1) = \col(S_2)$ because
$S_1, S_2 \subseteq T$, and $T$ was chosen so that all of its subsets of the
same size have the same color. From the definition, there is some index $j \in
[b]$ such that 
\[
\col(S_1)_j = Q'_i(\phi(x), \phi(y)) = Q_i(x,y) \qquad \text{and} \qquad \col(S_2)_j = Q'_i(\phi(u), \phi(v)) = Q_i(u,v),
\]
and thus $Q_i(x,y) = Q_i(u,v)$ as desired. Finally, suppose $(x,y)$ and $(u,v)$
do not have the same $D$-signature, though they must still have the same
$\overline{\Delta'}$-type, by assumption (meaning in particular $\Delta \neq
\emptyset$). Consider the pairs $(x', y')$  defined as
\begin{align*}
    \domino{x'}{y'} &= \domino{x_{i_1}}{y_{i_2}} \domino{x_{i_2}}{y_{i_2}} \dotsm \domino{x_{i_k}}{y_{i_k}}
        \; \domino{a}{a} \domino{a}{a} \dotsm \domino{a}{a} \\
    \domino{u'}{v'} &= \domino{u_{j_1}}{v_{j_2}} \domino{u_{j_2}}{u_{j_2}} \dotsm \domino{v_{j_k}}{v_{j_k}}
        \; \domino{a}{a} \domino{a}{a} \dotsm \domino{a}{a}
\end{align*}
where $i_1 < i_2 < \dotsm < i_k$ and $j_1 < j_2 < \dotsm < j_k$ are the indices of the non-$\domino{a}{a}$
dominoes of $(x,y)$ and $(u,v)$ respectively. Observe that $(x',y')$ has the same $D$-type and
$\overline{\Delta'}$-type as $(x,y)$, and $(u',v')$ has the same $D$-type and
$\overline{\Delta'}$-type as $(u,v)$.
This is because the order of non-$\domino{a}{a}$
dominoes is preserved, and the number of
non-$\domino{a}{a}$ dominoes in $(x,y)$ and $(u,v)$ is the same since they have
the same $\overline{\Delta'}$-type. By the argument above, we have $Q_i(x,y) =
Q_i(x',y')$ and $Q_i(u,v) = Q_i(u',v')$ for each query $Q_i$.
Finally, observe that $(x',y')$ and $(u',v')$ have the same
$\overline{\Delta}$-type, since by assumption they have the same
$(\overline{\Delta'}=\overline{\Delta \cup \domino{a}{a}})$-type (meaning the
order of non-$(\Delta \cup \domino{a}{a})$ dominoes is the same). Then $(x',y')$ is obtained from $(u',v')$
by swaps of consecutive dominoes involving dominoes in $\Delta$. Since each
query $Q_i$ is $\Delta$-shuffle invariant, we have $Q_i(x',y') = Q_i(u',v')$,
and therefore
\[
    Q_i(x,y) = Q_i(x',y') = Q_i(u',v') = Q_i(u,v) \,,
\]
as desired.
\end{proof}

Applying \cref{claim:main-step-1} twice, with $\Delta = \emptyset$ and $a=0$, and then with
$\Delta = \left\{\domino 0 0 \right\}$ and $a = 1$, we achieve
the $\left\{\domino 0 0, \domino 1 1\right\}$-shuffle property. We conclude with the following claim:

\begin{claim}
Suppose that the $\left\{\domino 0 0, \domino 1 1\right\}$-shuffle property holds. Then the $\left\{\domino 0 0, \domino 1 1, \domino 0 1, \domino 1 0\right\}$-shuffle property also holds.
\end{claim}
\begin{proof}[Proof of claim]
Since $\cQ$ is stable, there exists a constant $t$ such that neither $\GT_t$ or its complement $\neg \GT_t$ belong to $\QS(\cQ)$.

Take $\Delta= \left\{\domino{0}{1},\domino{1}{0}\right\}$. Take $N \in \bN$ such that $N-n > 2t$, and embed $\zo^n$ into $\zo^N$ by the map $\phi : x \mapsto x00\cdots0$, that pads $N-n$ many $0$'s at the end of the input string.
Let $Q_1', \dotsc, Q_c' \in \cQ$ be
 $\overline{\Delta}$-shuffle invariant query matrices such that $\EHD_k^N(X,Y) = f(Q_1'(X,Y), \dotsc, Q_c'(X,Y)) $  for all $X,Y \in \zo^N$. Now for each $i \in [c]$, define $Q_i : \zo^n \times \zo^n \to \zo$ as
$Q_i(x,y) = Q_i'(\phi(x),\phi(y))$.
Note that $\Gamma_\Delta(x,y) = \Gamma_\Delta(\phi(x), \phi(y))$.

Since each $Q'_i$ is $\overline \Delta$-shuffle invariant, $Q'_i(X,Y)$ depends only on the $\Delta$-type $\chi_\Delta(X,Y)$, and $Q_i(x,y)$
depends only on the $\Delta$-type $\chi_\Delta(x,y)$.
Therefore, for any $\Delta$-type $A$, we write $Q'_i(A)$ for the value taken by $Q'_i$ on all $(X,Y)$ with $\chi_\Delta(X,Y) = A$.

Assume for the sake of contradiction that there exists $i \in [c]$ such that $Q_i$ is not permutation-invariant.
Then there exist $x,y,u,v \in \zo^n$ such that $Q'_i(\phi(x),\phi(y)) \neq Q'_i(\phi(u),\phi(v))$, and, for the two $\Delta$-types $A = (\Gamma^A, \tau^A) = \chi_\Delta(\phi(x), \phi(y))$ and $B = (\Gamma^B, \tau^B) = \chi_\Delta(\phi(x), \phi(y))$:
\begin{enumerate}
    \item $Q'_i(A) \neq Q'_i(B)$; and
    \item The $\Delta$-signatures $\Gamma^A, \Gamma^B$ (subsequences of $\domino{1}{0}$ and $\domino{0}{1}$ dominoes) are permutations of each other, the tallies $\tau^A = \tau^B =: \tau$ are the same, and $\tau_{00} \geq N-n$ due to the padding in $\phi$.
\end{enumerate}
It suffices to consider $\Delta$-types with tally $\tau$ and signatures $\Gamma^A$ and $\Gamma^B$, where $\Gamma^B$ is obtained by swapping a single consecutive 
pair of dominoes in $\Gamma^A$;
if it holds that $Q'_i(A) = Q'_i(B)$ for any two $\Delta$-types with tally $\tau$ and $\Delta$-signatures $\Gamma^A, \Gamma^B$ which differ only by swapping a 
consecutive pair of dominoes, then it holds that $Q'_i(A) = Q'_i(B)$ for all $\Delta$-types $A, B$ with tally $\tau$, since $A$ may be transformed into $B$ by a sequence 
of swaps of consecutive dominoes.

Thus, assume $\Gamma^A$ and $\Gamma^B$ differ by only one swap of consecutive dominoes. 
Then we may choose domino sequences $\Sigma^A, \Sigma^B$ with the $\Delta$-types $A$ and $B$ as follows.

For some domino subsequences $\domino{C^{\circ}}{C_{\bullet}} \in \Delta^{d_1}$ and $\domino{D^{\circ}}{D_{\bullet}} \in \Delta^{d_2}$, where $d_1, d_2 \in \bN$ satisfy $d_1 + 2 + d_2 + \tau_{00} + \tau_{11} = N$, the following $\Sigma^A, \Sigma^B \in (\zo^2)^N$ have $\Delta$-types $A$ and $B$ respectively:
\begin{align*}
\Sigma^A &= \left( \domino{C^{\circ}}{C_{\bullet}}
\left(\domino{1}{0} \domino{0}{1} \right)
\domino{D^{\circ}}{D_{\bullet}}
\left(\domino{0}{0}\right)^{\tau_{00}}
\left( \domino{1}{1} \right)^{\tau_{11}} \right), \\
\Sigma^B &= \left( \domino{C^{\circ}}{C_{\bullet}}
\left(\domino{0}{1} \domino{1}{0} \right)
\domino{D^{\circ}}{D_{\bullet}}
\left(\domino{0}{0}\right)^{\tau_{00}}
\left( \domino{1}{1} \right)^{\tau_{11}} \right).
\end{align*}
To achieve a contradiction, we construct an impossibly large $\GT$ submatrix within $Q'_i$ as follows.
Let $e_i \in \zo^{\tau_{00}+2}$ be the string that has $1$ in the $i$th coordinate and is $0$ everywhere else. For $i \in [\tau_{00} +2]$, define $h_i \in \zo^N$ as
\[
  h_{i} = \begin{cases}
    C^{\circ} \mid e_{i} \mid D^{\circ}\mid\underbrace{1 \ldots 1 }_{\tau_{11}} &\text{ if } i \text{ is even,} \\ 
    C_{\bullet} \mid e_{i} \mid D_{\bullet} \mid \underbrace{1 \ldots 1}_{\tau_{11}} &\text{ otherwise.}
  \end{cases}
\]
 Now consider the $\lfloor\nicefrac{\tau_{00}}{2} \rfloor \times \lfloor\nicefrac{\tau_{00}}{2} \rfloor$ submatrix $M$ of $Q_i'$ on rows $h_{2i}$ and on columns  $h_{2j+1}$ for $i,j \in \left[\lfloor\nicefrac{\tau_{00}}{2} \rfloor \right]$. Note that if $i \leq j$, the pair $(h_{2i}, h_{2j+1})$ has $\Delta$-signature $\Gamma^A$, and if $i > j$, the pair has the $\Delta$-signature $\Gamma^B$. Since $Q_i'$ has different values on the input pairs that have $\Delta$-signatures $\Gamma^A$ and $\Gamma^B$, respectively, the submatrix $M$ is exactly $\GT_{\lfloor\nicefrac{\tau_{00}}{2}\rfloor}$ or its complement. Since $\tau_{00} \geq N-n > 2t$, this contradicts the fact that $\QS(\cQ)$ does not have $\GT_{T}$ for any $T > t$.
\end{proof}

\section{Main Results: Separations within \texorpdfstring{$\CCC$}{BPP0}}
\label{section:applications}

We now apply \cref{lem:permutation_invariant} to prove our main results.
The essence of our technique is that it transforms the task of proving a lower bound for a (total)
communication problem $\cP$ against an arbitrary constant number of oracle queries, into a lower
bound for a certain type of \emph{partial} subproblem of $\cP$ against a \emph{single} oracle query.
This type of partial problem is defined below.

\subsection{Reduction to One Query}
\label{section:reduction-to-one-query}
Recall that the \emph{tally} $\tau(x,y)$ of two $x,y \in \zo^n$ counts the
number of times each domino $\domino a b$ appears in the dominoes of $(x,y)$.

\begin{definition}[Two-Tally Matrix]
We say a partial matrix $M \in \{0,1,*\}^{t \times t}$ is a \emph{two-tally matrix} of
$\EHD_k$ if there exist $n \in \bN$ and $x^{(1)}, \dotsc, x^{(t)}, y^{(1)}, \dotsc, y^{(t)} \in
\zo^n$ which satisfy the following:
\begin{enumerate}
\item The submatrix of $\EHD_k^n$ on rows $x^{(i)}$ and columns $y^{(j)}$ is a completion of $M$;
\item If $M(x^{(i)}, y^{(j)}) = M(x^{(i')}, y^{(j')}) \neq *$ have the same
non-$*$ value in $M$, then $(x^{(i)}, y^{(j)})$ and $(x^{(i')}, y^{(j')})$
have the same tally $\tau(x^{(i)}, y^{(j)}) = \tau(x^{(i')}, y^{(j')})$, meaning
that $(x^{(i)},y^{(j)})$ and $(x^{(i')}, y^{(j')})$ are permutations of each
other.
\end{enumerate}
\end{definition}

Our main results will follow by the application of the next lemma.
\begin{lemma}
\label{lemma:two-tally}
Let $\cQ$ be any stable set of matrices, let $k$ be any constant, and let $M$ be a two-tally
matrix of $\EHD_k$. If $\mathsf{D}^\cQ(\EHD_k) = O(1)$ then there is $L \in \QS(\cQ)$
that is a completion either of $M$, or its Boolean negation $\neg M$.
\end{lemma}
\begin{proof}
Suppose $\mathsf{D}^\cQ(\EHD_k)=O(1)$. By \cref{lem:permutation_invariant}, there is a constant $c$ and a function $f$ such that for every $n$, there exist permutation-invariant matrices $Q_1,\ldots, Q_c\in \mathcal{Q}$ such that 
\begin{equation}\label{eq:2sig}
\forall x,y \in \zo^n :\qquad \EHD_k^n(x,y) = f(Q_1(x,y), \ldots, Q_c(x,y)).
\end{equation}
Consider now the two-tally matrix $M$ of $\EHD_k$ defined by the vectors $x^{(1)}, \dotsc, x^{(t)}, y^{(1)}, \dotsc, y^{(t)} \in \zo^n$, and let $i_0,j_0$, $i_1,j_1$ be such that $M(x^{(i_b)}, y^{(j_b)})=b$ for $b \in \zo$ (we may assume such pairs exist as otherwise the claim is trivial). By \eqref{eq:2sig}, there must exist $\ell$ such that $Q_\ell( x^{(i_0)}, y^{(j_0)}) \neq Q_\ell( x^{(i_1)}, y^{(j_1)})$, which by the permutation invariance of $Q_\ell$ implies that $Q_\ell$ distinguishes between $0$s and $1$s of $M$. Then the submatrix $L$ of $Q_\ell$, on rows $\{x^{(1)}, \dotsc, x^{(t)}\}$ and columns $\{y^{(1)}, \dotsc, y^{(t)}\}$, is a completion of either $M$ or $\neg M$.
\end{proof}

\subsection{No Complete Problem for \texorpdfstring{$\CCC$}{BPP0}}

We now prove a lower bound for $\EHD_k$ against queries of a general form, which will imply our
main \cref{thm:intro-no-complete-problem}. For convenience, we will
state the general result in terms of the VC dimension.

\begin{definition}[VC Dimension]
\label{def:vc-dimension}
    The VC dimension of a matrix $M \in \zo^{N \times N}$ is the largest number $d$ such that there
    are $d$ columns $y^{(1)}, \dotsc, y^{(d)}$ that are \emph{shattered}, meaning that for every $S \subseteq [d]$ there is a row $x^S$ such that    $M(x^S, y^{(i)}) = 1$ if and only if $i \in S$.
\end{definition}

\begin{remark}
For every problem $\cP \in \CCC$ there is a constant $d$ such
that the VC dimension of any $P \in \cP$ is at most $d$. 
If $\cM$ is a set of matrices with unbounded VC dimension (for example, if $\cM$ is the \textsc{Set-Disjointness} communication problem),
then $\QS(\cM)$ is the set of all matrices, meaning in particular
that $\mathsf{D}^\cM(\cP) = 1$ for all communication problems 
$\cP$.
\end{remark}

A simple argument shows that \emph{any} fixed total matrix $M$ appears as a two-tally matrix of
$\EHD_k$, for sufficiently large constant $k$. (This is not the
same thing as saying $\EHD_k$ has unbounded VC dimension --
for each constant $k$, the VC dimension of $\EHD_k$ remains bounded.)

\begin{proposition}
\label{prop:two-sig-vc}
For every constant $k$, there is a $2^k \times k$ matrix $M$ of VC dimension $k$ that is a two-tally
submatrix of $\EHD_{k-1}$.
\end{proposition}
\begin{proof}
    Let $n > 2k$, and
    let $y^{(1)}, \dotsc, y^{(k)} \in \zo^n$ be the first $k$ standard basis vectors. Now for every set $S \subseteq [k]$ let $x^S \in \zo^n$ be the string where the last $k-|S|$ bits are set to 1, and the bits $i \in S$ are set to 1, and the remaining bits are 0. Fix $i$ and $S$ and consider two cases:
    \begin{itemize}
        \item  Suppose $i \in S$ and consider the domino sequence
    of $x^S, y^{(i)}$; we see that it contains 1 of $\domino{1}{1}$, $k-1$ of $\domino{1}{0}$, and $n-k$ of $\domino{0}{0}$, and $\dist(x^S,y^{(i)}) = k-1$.

    \item Now suppose $i \notin S$ and consider the domino sequence of $x^S, y^{(i)}$; we see that it contains
    $1$ of $\domino{0}{1}$, $k$ of $\domino{1}{0}$, and $n-k-1$ of $\domino{0}{0}$, and $\dist(x^S,y^{(i)}) = k+1$.
    \end{itemize}
We see that for every set $S \subseteq [k]$, $\EHD_{k-1}^n(x^S,y^{(i)}) = 1$ iff $i \in S$ and therefore this
$k \times 2^k$ submatrix has VC dimension $k$. From the above observations, this submatrix satisfies the conditions to
be a two-tally submatrix of $\EHD_{k-1}$.
\end{proof}

As a result, we get a general separation of $\EHD_k$ against oracle queries belonging to any stable set of
matrices. Note that any stable set of matrices has constant VC dimension, because a forbidden $\GT_t$
subproblem implies a bound of $t$ on the VC dimension.

\begin{theorem}
\label{thm:vc}
Let $\cQ$ be any stable set of matrices with VC dimension $d$. Then for any $k \geq d$, $\mathsf{D}^\cQ(\EHD_k) = \omega(1)$.
\end{theorem}
\begin{proof} Assume  $\mathsf{D}^\cQ(\EHD_k) = O(1)$.
    By \cref{prop:two-sig-vc} there is a $2^{k+1} \times (k+1)$ two-tally matrix $M$ of $\EHD_k$, with VC
    dimension $k+1 > d$. Note that $M$ and $\neg M$ are permutations of each other (meaning $\neg M$ is obtained from $M$ by permuting its rows and columns), so $M, \neg M \notin \QS(\cQ)$. This contradicts \cref{lemma:two-tally}. 
\end{proof}

Our main \cref{thm:intro-no-complete-problem} now follows as a corollary, since any problem
$\cQ \in \CCC$ must be stable (\cref{prop:ccc-stable}).

\begin{corollary}[\cref{thm:intro-no-complete-problem}]
For every problem $\cQ \in \CCC$, there exists a constant $k$ such that $\mathsf{D}^\cQ(\EHD_k) =
\omega(1)$.
\end{corollary}

This statement also holds for $\THD_k$ in place of $\EHD_k$, since
they are equivalent under constant-cost reductions. $\THD_k$ has
a one-sided error protocol (unlike $\EHD_k$) which also means there is no complete problem for the class of constant-cost problems with one-sided error.

\subsection{The \texorpdfstring{$k$}{k}-Hamming Distance Hierarchy}

The VC dimension of \textsc{Equality} is 1, since $\left[\begin{smallmatrix}
    1 & 1 \\ 0 & 1 \end{smallmatrix}\right]$ cannot occur as a submatrix. Therefore
    \[
        \mathsf{D}^{\EQ}(\EHD_1) = \omega(1) \,,
    \]
by \cref{thm:vc}, recovering (qualitatively) the results of \cite{HWZ22,HHH23eccc}.
An immediate consequence of \cref{thm:vc} is an infinite number of such separations, forming an
infinite hierarchy within $\CCC$.

\begin{corollary}[\cref{thm:intro-hierarchy}]
\label{thm:k-hd-hierarchy}
There are infinitely many $k \in \bN$ such that $\sD^{\EHD_k}(\EHD^n_{k+1})=\omega(1)$.
\end{corollary}
\begin{proof}
Fix any $t$. Then $\EHD_t$ is stable, so by \cref{thm:vc},
there is some constant $t' > t$ such that $\mathsf{D}^{\EHD_t}(\EHD_{t'}) = \omega(1)$. Then there must
exist $t \leq k < t'$ such that $\mathsf{D}^{\EHD_k}(\EHD_{k+1}) = \omega(1)$, because, if
$\mathsf{D}^{\EHD_k}(\EHD_{k+1}) = O(1)$ for every $t \leq k < t'$ then we would have
$\mathsf{D}^{\EHD_t}(\EHD_{t'}) = O(1)$.
\end{proof}

In the above proof, it suffices to take $t' = 2^{\Theta(t \log t)}$. First observe that any $M \in \QS(\EHD_t)$ has 
$\mathsf{R}(M) \leq  C \cdot t \log t$ for some constant $C$. On the other hand, we may choose
a $t' \times t'$ two-tally matrix $M \in \QS(\EHD_{t'})$ with maximum randomized communication cost
$\mathsf{R}(M) = \Theta(\log t')$, since \cref{prop:two-sig-vc} guarantees
that every $t' \times t'$ matrix with unique columns exists as a two-tally submatrix of $\EHD_{t'}$. Therefore, if we choose $t' = 2^{\Theta(t \log t)}$ with a sufficiently large constant in the exponent, then we have a two-tally submatrix $M$ of $\EHD_{t'}$ with $\mathsf{R}(M) = \Theta(\log t') > C \cdot t \log t$, so $M \notin \QS(\EHD_t)$ (and the same holds for $\neg M$).

However, the VC dimension argument does not suffice to separate $\EHD_1$ from $\EHD_2$, because the
VC dimension of $\EHD_1$ is 3 (the first 3 standard basis vectors are shattered), which only proves
$\mathsf{D}^{\EHD_1}(\EHD_3) = \omega(1)$. We tighten this separation by choosing a different
two-tally matrix $M$.

\begin{theorem}[Restatement of \cref{thm:intro-2-vs-1}]
$\mathsf{D}^{\EHD_1}(\EHD_2) = \omega(1)$.
\end{theorem}
\begin{proof}
First consider the matrix $M$ defined as the submatrix of $\EHD_2^7$ on rows $\cX$ and
columns $\cY = \cY_0 \cup \cY_1$, where
\begin{align*}
\cX &\define \{0011000,1100000\} \\
\cY_0 &\define \{0000011,0000101,0000110\} \\
\cY_1 &\define \{1010000,1001000,010100\}  \,.
\end{align*}
Observe that for $\beta \in \zo$,  $\cX \times \cY_\beta$ is a $\beta$-monochromatic rectangle of $\EHD_2^7$, and that the
distance between any two $(x,y) \in \cX \times \cY$ is either 2 or 4. Now, observe that for any two
distinct strings $a,b \in \zo^7$ with Hamming weight 2, there exists a $\delta = \delta(a,b) \in
\zo^7$ with Hamming weight 2 such that $\dist(a,\delta) = 2$ and $\dist(b,\delta) = 4$. We now
extend $M$ to a partial matrix $M'$ by adding the columns $\{ \delta(a,b) \,|\, a, b \in \cX, a \neq
b \}$ and the rows $\{ \delta(a,b) \,|\, a,b \in \cY, a \neq b \}$, and taking the entries
\[
  M'(x,y) \define \begin{cases}
    1 &\text{ if } \dist(x,y) = 2 \\
    0 &\text{ if } \dist(x,y) = 4 \\
    * &\text{ otherwise.}
  \end{cases}
\]
This matrix agrees with $\EHD^7_2(x,y)$ whenever $\dist(x,y) \in \{2,4\}$, and every row and every
column is distinct. Since the weight of every string is the same and all non-$*$ entries have
distance 2 or 4, it must hold that all 1-valued entries $x,y$ of $M'$ are domino permutations of
each other, and the same for all 0-valued entries, making $M'$ a two-tally matrix of $\EHD_2$.

Assume for the sake of contradiction that $\mathsf{D}^{\EHD_1}(\EHD_2) = O(1)$. Then by
\cref{lemma:two-tally}, there is $L \in \QS(\EHD_1)$ that is a completion of $M'$ or $\neg M'$. However,
$M'$ and $\neg M'$ both contain the submatrix $K_{2,3}$ (the 2$\times$3 all-1s matrix), and it is known that $\EHD_1$ does not contain
$K_{2,3}$. Since $M'$ has distinct rows and columns, it cannot be obtained from a submatrix of
$\EHD_1$ by copying rows and columns, and therefore any completion $L$ of $M'$ or $\neg M'$ cannot
belong to $\QS(\EHD_1)$, a contradiction.
\end{proof}

\subsection{Separating \texorpdfstring{$k$}{k}-Hamming Distance and Integer Inner Product}

Our final application separates \textsc{$k$-Hamming Distance} from \textsc{Integer Inner
Product}:

\begin{definition}[Integer Inner Product]
\label{def:iip}
For any fixed constant $d$, the \textsc{Integer Inner Product} problem in dimension $d$ is defined
as $\IIP_d = (\IIP_d^n)_{n \in \bN}$ where $\IIP_d^n : \zo^{dn} \times \zo^{dn} \to \zo$ is the
function defined on $x,y \in \zo^{dn}$, interpreted as the binary representation of integer vectors $x =
(x_1, x_2, \dotsc, x_d)$ and $y = (y_1, y_2, \dotsc, y_d)$ in domain $[-2^{n-1}, 2^{n-1}]^d$, where
\[
  \IIP_d^n(x,y) \define \begin{cases}
    0 &\text{ if } \inn{x,y} = 0 \\
    1 &\text{ otherwise.}
  \end{cases}
\]
\end{definition}
It is known that $\mathsf{R}(\IIP_d^n) = O(d \cdot \log n)$ \cite{CLV19}, so $\IIP_d \in \BPP$ for every
constant $d$, but it is conjectured that $\IIP_d \notin \CCC$ (see \eg \cite{CHHS23}). It was shown
in \cite{CLV19} that there is an infinite sequence $d_1 < d_2 < \dotsm$ of constants such that
$\mathsf{D}^{\IIP_{d_i}}(\IIP_{d_{i+1}}) = \Theta(n)$; in other words, they form an infinite
hierarchy within $\BPP$. We will show that oracles to these functions, which each have much larger
randomized communication complexity than any $\EHD_k$, nevertheless cannot simulate $\EHD_k$ in $\CCC$.

We require the following lemma.
\begin{lemma}
For any constant $d$, $\IIP_d$ is stable and has VC dimension at most $d$.
\end{lemma}
\begin{proof}
It suffices to prove the following claim.

\begin{claim}
Let $t \geq 1$, and let $\cX, \cY \subseteq \bR^t$ be any finite sets of points with $\vec 0 \notin \cY$, and consider any sequences of points $x_1, \dotsc, x_m \in \cX$ and $y_1, \dotsc, y_m \in \cY$
such that $\forall i,j \in [m]$, $\inn{x_i,y_j} = 0$ if and only if $i \leq j$. Then $m \leq t$.
\end{claim}
\begin{proof}[Proof of claim]
We prove the claim by induction on $t$. One may easily check that the claim is true in the base case $t=1$ where
$\inn{x_i,y_j} = 0$ iff $x_i = 0$. Now assume $t \geq 2$. Let $x_1, \dotsc,
x_m \in \cX$ and $y_1, \dotsc, y_m \in \cY$ be sequences satisfying the condition $\inn{x_i, y_j} =
0$ iff $i \leq j$. Since $y_m \neq \vec 0$ and $\inn{x_i,y_m} = 0$ for all $i \in [m]$, it defines a
perpendicular subspace $\cW$ of dimension $t-1$, $\cW \define \{ x \in \bR^t : \inn{x,y_m} = 0 \}$,
such that $\{x_1, \dotsc, x_m\} \subseteq \cW$. Let $y'_1, \dotsc, y'_{m-1}$ denote the projections of $y_1, \dotsc,
y_{m-1}$ into $\cW$, and observe for each $j \in [m-1]$ that $\inn{x_m,y_j} \neq 0$ by definition, so $y'_j
\neq \vec 0$ since $x_m \in \cW$. Finally note that for all $i,j \in [m-1]$, it holds that $\inn{x_i,y'_j} = 0$ iff
$\inn{x_i, y_j}=0$, and therefore we may apply the induction hypothesis to $x_1, \dotsc, x_{m-1}$ and
$y'_1, \dotsc, y'_{m-1}$ to conclude that $m-1 \leq t-1$.
\end{proof}

To conclude the proof of the lemma, observe that taking finite sets $\cX, \cY \subseteq \bR^d$ that
may include $\vec 0$ cannot increase the size $m$ of the ordered sequences $x_1, \dotsc, x_m$ and
$y_1, \dotsc, y_m$ in the claim by more than 1. The bound on the VC dimension is known in the literature, and also
follows from the above proof, since this exhibits a $(d+1) \times (d+1)$ matrix with unique rows and columns that cannot be a submatrix of $\IIP_d$.
\end{proof}

We now apply \cref{thm:vc} to separate $\EHD_k$ from $\IIP_d$ and obtain the formal statement of
\cref{thm:intro-iip}.

\begin{theorem}
\label{thm:formal-iip}
For any constant $d$ and any $k \geq d$, $\mathsf{D}^{\IIP_d}(\EHD_k) = \omega(1)$.
\end{theorem}

There are two conjectures in the literature which would imply
a stronger version of this theorem for any constant $d$ and $k=1$;
see \cref{section:discussion}. Note that $\sD^{\IIP_2}(\EHD_0) = O(1)$ since $\EHD_0$ is \textsc{Equality}.

\section{The Complexity of Submatrices of \texorpdfstring{$k$}{k}-Hamming Distance}
\label{section:khd-submatrices}

Problems in $\CCC$ satisfy the hereditary property that the randomized
communication cost $\mathsf{R}(\cdot)$ remains bounded by taking submatrices.
As discussed in \cref{section:intro-reductions}, it is not true in general (outside $\CCC$) that 
$\mathsf{R}( \cdot )$ is preserved, as a function of the matrix size, by taking
arbitrary submatrices. It is helpful,
for proving lower bounds, to understand when hereditary properties
hold for other complexity measures as well. We will show in this section that 
$\sD^\EQ(\EHD_k)$, and the $\gamma_2$-norm of $\EHD_k$, are also preserved when
taking submatrices of $\EHD_k$, and
we use these hereditary properties to prove new lower bounds against $\EHD_k$ oracles
within $\BPP$.

The $\gamma_2$-norm is an important norm in
communication complexity. For a matrix $M \in \zo^{N \times N}$ it is defined as
\[
    \|M\|_{\gamma_2} = \min_{A,B : \; M=AB}
        \|A\|_{\text{row}} \|B\|_{\text{col}} \,,
\]
where $\|A\|_{\text{row}}, \|B\|_{\text{col}}$ denote
the largest $\ell_2$ norm of any row and column, 
respectively. See \cite{LS09,HHH23eccc,CHHS23} for a 
discussion of this quantity and its relation to 
communication complexity.

We prove the following proposition in 
\cref{prop:submatrix-protocol}.

\begin{proposition}\label{prop:eqprotocol-for-kHD}
For any $k$ and any $N \times N$ query matrix $Q \in \QS(\THD_k)$, we have 
\begin{enumerate}
    \item $\sD^{\EQ}(Q) = O(k \log \log N)$, and
    \item $ \|Q\|_{\gamma_2} = (\log N)^{O(k)}$. 
\end{enumerate}
\end{proposition}
Item (2) above follows from Item (1) by the result of \cite{HHH23eccc} that for any Boolean matrix $M$, \begin{equation}\label{eq:eqvsgamma2}
\sD^{\EQ}(M) \geq \frac{1}{2} \log \|M\|_{\gamma_2}.
\end{equation}
It was proved later in \cite{CHHS23} that for any $d\geq 3$, 
\begin{equation}\label{eq:gamma2ofIIP}
\|\IIP_d^n\|_{\gamma_2} = 2^{\Omega(n)}.
\end{equation}
This combined with \eqref{eq:eqvsgamma2} recovers the lower-bound of \cite{CLV19}, showing that for every $d\geq 3$,
\begin{equation}\label{eq:IIPLB}
\sD^{\EQ}(\IIP_d^n)= \Omega(n).
\end{equation}

Combining this with \cref{prop:eqprotocol-for-kHD} shows that even $\IIP_3^n$ does not reduce to 
the \textsc{$k$-Hamming Distance} under $\BPP$ reductions allowing queries of unbounded size, for any
$k \leq n / (\log n)^{\omega(1)}$:
\begin{theorem}[Restatement of \cref{thm:intro-bpp}]
\label{thm:IIP-to-khd}
For any $d \geq 3$ and $k\leq \nicefrac{n}{(\log n)^{\omega(1)}}$, we have
$\sD^{\THD_k}(\IIP_d^n)=(\log n)^{\omega(1)}$. In particular, when $k$ is a constant, $\sD^{\THD_k}(\IIP_d^n)=\Omega\left(\nicefrac{n}{\log n}\right)$.
\end{theorem}
\begin{proof}
Consider a $\sD^{\THD_k}$ protocol for $\IIP_d^n$ with $q$ queries. By \cref{prop:eqprotocol-for-kHD}, each of the queries is a $2^{dn}\times 2^{dn}$ matrix $Q\in \QS(\THD_k)$ that can be simulated with $O(k \log n)$ \textsc{Equality} oracle queries. This gives a protocol for $\IIP_d^n$ with $O(qk\log n)$ \textsc{Equality} queries.
Using \eqref{eq:IIPLB}, we have $q=\Omega(\nicefrac{n}{k \log n})$. Therefore, $q=(\log n)^{\omega(1)}$, as long as $k\leq n / (\log n)^{\omega(1)}$, and $q=\Omega(\nicefrac{n}{\log n})$ when $k$ is a constant. 
\end{proof}

\subsection{Replacing \texorpdfstring{$k$}{k}-Hamming Distance Queries with \textsc{Equality} Queries}
\label{prop:submatrix-protocol}

We now show that any $N\times N$ matrix $Q\in \QS(\THD_k)$ can be reduced to $O(k\log \log N)$ \textsc{Equality} oracle calls. When $Q=\THD_k^n$, \ie when the matrix has size $2^n \times 2^n$ and the inputs are $x,y \in \zo^n$, this can be done easily using binary
search to find the first differing bit, and removing it and repeating up to $k+1$ times. But this simple protocol becomes inefficient when $Q$ is a submatrix of $\THD_k^d$ for $d \gg \log N$, \ie the inputs
$x,y$ are chosen from subsets $X, Y \subseteq \zo^d$ with $|X|,|Y| \ll 2^d$, so the number of coordinates is very large and na\"ive binary search gives $O(k \log d)$ instead of $O(k \log \log N)$. Such a $d$-dependent bound is not useful for our purposes, as our protocols are allowed to query oracles of arbitrary dimension. 

We will need the following simple lemma. For a binary string $x\in \{0,1\}^d$ and a set $A\subseteq [d]$, write $x_A\in\{0,1\}^{|A|}$ for the substring of $x$ on indices $A$.

\begin{lemma}
\label{prop:diameter-partition}
Let $Z \subseteq \zo^d$ be a set of $N$ binary strings. Then there exists a partition $[d] = A \cup B$ such that, for all $x,y \in Z$:
\begin{enumerate}
    \item If $x_A = y_A$, then $x_B, y_B$ differ on at most $3 \log N$ bits.
    \item If $x_B = y_B$, then $x_A, y_A$ differ on at most $3 \log N$ bits.
\end{enumerate}
\end{lemma}
\begin{proof}
Choose a partition $[d]=A\cup B$ uniformly at random. Fix an arbitrary pair $x,y\in Z$. If $x$ and $y$ differ on at most $3\log N$ bits, then regardless of $A$ and $B$, the properties hold for $x,y$. Assume otherwise, and note that in this case, the probability that $x_A=y_A$ or $x_B=y_B$ is at most $2\cdot 2^{-3\log N}<1/{N \choose 2}$.  Thus, by a union bound over all the ${N \choose 2}$ choices of $x,y$, there exists a choice of $A$ and $B$ that satisfies the claim. 
\end{proof}

We also require the next proposition, which is well-known and is achieved by performing binary search
on the bits in the binary representation of the inputs $i,j \in [N]$.

\begin{proposition}
\label{prop:gt-protocol-eq}
    $\mathsf{D}^{\EQ}(\GT_N) = O(\log\log N)$.
\end{proposition}

Now we show that any $N\times N$ submatrix of $\THD_k$, with arbitrary dimension, can be computed efficiently 
using $\EQ$ oracles.

\begin{proof}[Proof of \cref{prop:eqprotocol-for-kHD}]
Let $Q \in \QS(\THD_k)$, and let $X, Y \subseteq \zo^d$ be the sets of rows and columns of $Q$ respectively, for some dimension $d$. Write $Z = X \cup Y$ for the set of all relevant binary strings. In the statement of \cref{prop:eqprotocol-for-kHD} we have defined $N=|X|=|Y|$ but in the proof here we write $N = |Z|$ for convenience, which does not affect the conclusion.
We first define the following procedure $\textsc{Bounded Diameter Threshold Distance}$, which uses \textsc{Equality} oracle queries to compute the Hamming distance (up to threshold $k$), on inputs that are promised to belong
to a set of $N$ inputs of diameter at most $3\log N$. This protocol works by transforming a low-diameter set into a low-Hamming-weight set. This subroutine will be used in $\textsc{Threshold Distance}$ protocol determining whether $\dist(x,y)>k$.

\definecolor{CommentColor}{HTML}{666666}
\newcommand{\pcomment}[1]{{\color{CommentColor}\;\;$\triangleright$ \textit{#1}}}
\begin{algorithm}[H]
    \floatname{algorithm}{\textsc{Bounded Diameter Threshold Distance}($Z,x,y,k$)}
    \renewcommand\thealgorithm{}
    \caption{\ } 
    \begin{algorithmic}
        \State\textbf{Requires} $x,y \in Z$, $\forall u,v \in Z, \dist(u,v) \leq 3 \log N$. \pcomment{Write $N = |Z|$.}
        \State Alice and Bob, without communication, determine the following:
        \State \qquad Pick the lexicographically first $z\in Z$.
        \State \qquad $I\leftarrow \{i \ | \ \exists w\in Z, (w\oplus z)_i=1\}$. \pcomment{Note that $|I|=O(N\log N)$.}
        \State Alice lets $S\leftarrow \{i \ | \ (x \oplus z)_i=1\}$.
        \State Bob lets $T\leftarrow \{j \ | \ (y\oplus z)_j=1\}$. \pcomment{Note that $S,T\subseteq I$ and $|S|,|T|\leq 3\log N$}.
        \State Initialize $c\leftarrow 0$. \pcomment{The number of differing bits that are confirmed so far.}
        \While{$c<k$ and $\EQ(S,T) = 0$} 
            \State $c\leftarrow c+1$.
            \State Initiate $S'\leftarrow S$ and $T'\leftarrow T$,
                and $K = \lceil 3 \log N \rceil$.
            \While{$K>1$}
                \State\pcomment{Here we determine the smallest element in the symmetric difference of $S'$ and $T'$.}
                \State $K \gets \lceil K/2 \rceil$. 
                \State $S_1\leftarrow $ the first $\min\{|S'|,  K \}$ elements of $S'$. $S_2\leftarrow S' \setminus S_1$.
                \State $T_1\leftarrow $ the first $\min\{|T'|, K \}$ elements of $T'$. $T_2\leftarrow T' \setminus T_1$.
                \If{$\EQ(S_1, T_1) = 1$} $S'\leftarrow S_1$, $T'\leftarrow T_1$.
                \Else{}  $S' \leftarrow S_2, T' \leftarrow T_2$.
                \EndIf
            \EndWhile

            \If{$|S'|=0$} 
                \State Let $j$ be such that $T'=\{j\}$. Bob lets $T\leftarrow T-\{j\}$.
            \ElsIf{$|T'|=0$}
                \State Let $i$ be such that $S'=\{i\}$. Alice lets $S\leftarrow S-\{i\}$.
            \Else
                \State Let $i,j$ be such that $S'=\{i\}$ and $T'=\{j\}$ \pcomment{Alice knows $i$ and Bob knows $j$.}
                \State Alice and Bob determine whether $i < j$ using \cref{prop:gt-protocol-eq} on domain $I$.
                \If{$i<j$} Alice lets $S\leftarrow S-\{i\}$. 
                \Else{} Bob lets $T\leftarrow T-\{j\}$.
                \EndIf
            \EndIf
        \EndWhile
        \If{$S=T$} 
            \Return $c$.
        \Else\ 
            \Return $\bot$.
        \EndIf
    \end{algorithmic}
\end{algorithm}

\begin{claim}
Let $k \in \bN$ and $Z \subseteq \zo^d$ be shared inputs to both parties, where $Z$ satisfies $|Z|=N$ and $\dist(u,v) \leq 3\log N$ for all $u,v \in Z$. Then on inputs $x,y \in Z$, the protocol
$\textsc{Bounded Diameter Threshold Distance}(Z,x,y,k)$ uses at most $O(k \log\log N)$ calls to the \textsc{Equality} oracle and outputs the following:
\begin{itemize}
    \item If $\dist(x,y) \leq k$, the protocol outputs $\dist(x,y)$. 
    \item Otherwise the protocol outputs $\bot$.
\end{itemize}
\end{claim}

The analysis of the number of \textsc{Equality} oracle calls is elementary.
The correctness of the \textsc{Bounded Diameter Threshold Distance} protocol follows from the fact that the
inner loop satisfies the following invariant. Let $i \in \bN$ be the smallest element in the symmetric difference
of $S', T'$. Then
\begin{itemize}
    \item If $S_1 \neq T_1$ then $i \in S_1 \cup T_1$, and it is the smallest element in the symmetric difference
    of those two sets.
    \item If $S_1 = T_1$ then $i \in S_2 \cup T_2$, and it is the smallest element in the symmetric difference
    of those two sets.
\end{itemize}

We now use the bounded-diameter search sub-protocol to construct the full search protocol. 

\begin{algorithm}[H]
    \floatname{algorithm}{\textsc{Threshold Distance}($x,y,k$)}
    \renewcommand\thealgorithm{}
    \caption{}
    \begin{algorithmic}[H]
        \If{$\EQ(x,y) = 1$} \Return $0$.
        \ElsIf{$k=0$} \Return $\bot$.
        \ElsIf{$k>0$}
        \State Partition $[d] = A \cup B$ according to \cref{prop:diameter-partition} applied with $Z=X\cup Y$, where $N\leq 2^{n+1}$.
    \If{$\EQ(x_A,y_A)=1$ and $\EQ(x_B, y_B)=0$}
            \State $Z_B\leftarrow \{ z_B \ | z\in Z \wedge \ z_A= x_A = y_A\}$.
            \State \Return \Call{Bounded Diameter Threshold Distance}{$Z_B, x_B, y_B, k$}.
            \State \pcomment{Returns the correct value due to \cref{prop:diameter-partition} and \cref{claim:bd-diam-distance}.}
        \ElsIf{$\EQ(x_B, y_B) = 1$ and $\EQ(x_A, y_A) = 0$}
            \State $Z_A\leftarrow \{ z_A \ | z\in Z \wedge \ z_B= x_B = y_B\}$.
            \State \Return \Call{Bounded Diameter Threshold Distance}{$Z_A, x_A, y_A, k$}.
            \State \pcomment{Returns the correct value due to \cref{prop:diameter-partition} and \cref{claim:bd-diam-distance}.}
        \Else \ \pcomment{In this case $\dist(x_A,y_A), \dist(x_B,y_B) \geq 1$.}
            \State $t \leftarrow $ \Call{Threshold Distance}{$x_A, y_A, k-1$}.
                \State\pcomment{Returns $t=\dist(x_A,y_A)$ if $\dist(x_A,y_A)\leq k-1$.}
            \If{$t=\bot$} \Return $\bot$. \pcomment{$\dist(x_A,y_A) + \dist(x_B,y_B) > (k-1) + 1$.}
            \EndIf
            \State $r\leftarrow $ \Call{Threshold Distance}{$x_B, y_B, k-t$}.
            \State \pcomment{Returns $r = \dist(x_B,y_B)$ if $\dist(x_B,y_B)\leq k-\dist(x_A,y_A)$.}
            \If{$r=\bot$} \Return $\bot$.
            \Else{} \Return $t+r$.
            \EndIf     
        \EndIf
        \EndIf
    \end{algorithmic}
\end{algorithm}

The proposition
follows immediately from the next claim.

\begin{claim}
\label{claim:bd-diam-distance}
    Let $k \in \bN$ be a shared input to both parties. Then on inputs $x \in X, y \in Y$,
    the protocol $\textsc{Threshold Distance}(x,y,k)$ uses at most $O(k \log \log N)$ calls to the \textsc{Equality}
    oracle and outputs the following:
    \begin{itemize}
        \item If $\dist(x,y) \leq k$, the protocol outputs $\dist(x,y)$.
        \item If $\dist(x,y) > k$, the protocol outputs $\bot$.
    \end{itemize}
\end{claim}

The correctness of the \textsc{Threshold Distance} protocol follows from the claims in the comments, and
the observation that the number of \textsc{Equality} oracle queries is $O(k \log \log N)$, which can be computed by an elementary recurrence.
\end{proof}

\section{Discussion and Open Problems}
\label{section:discussion}

Our \cref{thm:intro-hierarchy} shows that there is an infinite hierarchy of \textsc{$k$-Hamming
Distance} problems within $\CCC$ that are irreducible to lower levels of the hierarchy. We expect
it to be the case that $\mathsf{D}^{\EHD_k}(\EHD_{k+1}) = \omega(1)$ for every constant $k$. Indeed, it seems natural to
expect that, for inputs on $n$ bits,
$\mathsf{D}^{\EHD_k}(\EHD_{k+1}) = \Omega(\log n)$, which matches an easy binary search based upper bound of $O(\log n)$. This was proved for $k=0$ in \cite{HHH23eccc}. It is possible that the question could be answered using the
technique of \cite{CLV19} combined with an analysis of monochromatic rectangles in $\EHD_k$.

\begin{question}
Is it the case that $\mathsf{D}^{\EHD_k}(\EHD_{k+1}^n) = \Theta(\log n)$?
\end{question}

One question that arose in the course of this work is whether a certain \emph{dimension reduction}
result holds for \textsc{$k$-Hamming Distance}. Constant-cost reductions (\cref{def:reductions}) to the \textsc{$k$-Hamming Distance} problem $\EHD_k$ allow queries of arbitrarily large dimension. The
question is whether any such query can be replaced with a constant number of queries to $\EHD_k$
with dimension $O(\log N)$ where $N$ is the original domain size.
Formally: 

\begin{question}
Let $M \in \zo^{N \times N}$ be an arbitrary submatrix of $\EHD^n_k$ where $n$ is arbitrarily large.
Is there an absolute constant $c$ and a function $f : \zo^c \to \zo$ such that
\[
  \forall i,j \in [N] :\qquad M(i,j) = f(H_1(i,j), H_2(i,j), \dotsc, H_c(i,j)) \,,
\]
where each $H_i$ is an $N \times N$ submatrix of $\EHD^d_k$ with $d = O(\log N)$?
\end{question}

If this question has a positive answer, it may be helpful in the future for lower bounds against
Hamming Distance oracles; it is one strategy that we tried in pursuit of \cref{thm:intro-hierarchy}.

An important question left open by this paper is whether
the \textsc{$k$-Hamming Distance} captures the entirety of $\CCC$, up to reductions. In other words, for every problem $\cP \in \CCC$, there exists a
constant $k$ such that $\mathsf{D}^{\EHD_k}(\cP) = O(1)$. We do not believe this to be the case, but there is not any
example of a problem $\cP \in \CCC$ in the literature that might serve as a candidate
counterexample.

Finally, we point out two conjectures in the literature that would imply a stronger form of \cref{thm:formal-iip} that
holds for any constant $d$ and $k = 1$. The first conjecture is that, if 
$\cQ$ and $\cP$ are any problems where $\mathsf{D}^\cQ(\cP) = O(1)$ and $\cQ$ has bounded sign-rank (which holds in
particular for $\IIP_d$ \cite{CHHS23}), $\cP$ also has bounded sign-rank \cite{HHPTZ22}. The second
conjecture is that $\EHD_1$ has unbounded sign-rank \cite{HHPTZ22}.

\begin{center}
    \textbf{\Large Acknowledgments}
\end{center}
Thanks to Viktor Zamaraev for many helpful discussions and for helpful comments on the presentation of the main
proof, and thanks to the anonymous reviewers for their comments which helped improve the presentation of this article.

\bibliographystyle{alpha}
\bibliography{references.bib}

\end{document}